%% file: main.tex
\documentclass[reqno]{amsart}

\usepackage{hyperref}
\usepackage[foot]{amsaddr}
\usepackage[table]{xcolor}
\usepackage[maxbibnames=99]{biblatex}
\DeclareFieldFormat{note}{}
\DeclareFieldFormat[misc]{title}{{``#1''}}
\DeclareFieldFormat[misc]{titlecase}{\MakeSentenceCase*{#1}}
\DeclareFieldFormat[inproceedings]{titlecase}{\ifthenelse{\ifcurrentfield{booktitle}}{#1}{\MakeSentenceCase*{#1}}}
\renewbibmacro*{volume+number+eid}{%
  \printfield{volume}%
  \setunit*{\addcolon}%
  \printfield{number}%
  \setunit{\addcomma\space}%
  \printfield{eid}}
\usepackage{xparse}
\usepackage{subcaption}
\usepackage{tikz}
\usepackage[all]{xy}
\usepackage[small]{complexity}
\usepackage{colortbl,tabularx}
\usetikzlibrary{automata}
\usetikzlibrary{positioning}
\usetikzlibrary{calc}
\bibliography{journalsabbr,bib}
\usepackage{url}
\usepackage[capitalize]{cleveref}

\usepackage{amsmath,amssymb,stmaryrd,thmtools}
\input{drawings.tex}

\def\keywords{\smallskip\noindent\textsc{Keywords.}~~}

\newcommand{\myqed}{}

\title[One-dimensional Pushdown VASS with Resets]{Coverability
  is Undecidable in One-dimensional Pushdown Vector Addition Systems with Resets}
\author{Sylvain Schmitz$^{1,2}$ \and Georg Zetzsche$^3$}
\address{$^1$~LSV, ENS Paris-Saclay \& CNRS
\\Universit\'e Paris-Saclay
\\France
}
\address{$^2$~IUF\\France}
\address{$^3$~Max Planck Institute for Software Systems (MPI-SWS)\\Germany}

\newcommand{\izero}{\mathsf{z}}
\newcommand{\ireset}{\mathsf{r}}
\newcommand{\itransfer}{\mathsf{t}}
\newcommand{\iinc}{\mathsf{\mathord{+}}}
\newcommand{\idec}{\mathsf{\mathord{-}}}

\newcommand{\eqdef}{\eqby{def}}
\newcommand{\eqby}[1]{\mathrel{\raisebox{-.1ex}{\ensuremath{\stackrel{\raisebox{-.25ex}{\scalebox{.45}{\upshape\textrm{#1}}}}{=}}}}}

\newcommand{\N}{\mathbb{N}}
\newcommand{\Nctr}{\N}
\newcommand{\NctrReset}{\N_{\ireset}}
\newcommand{\NctrTransfer}{\N_{\itransfer}}
\newcommand{\NctrZero}{\N_{\izero}}
\newcommand{\PD}{\mathsf{PD}}
\newcommand{\cM}{\mathcal{M}}
\newcommand{\cV}{\mathcal{V}}
\newcommand{\cA}{\mathcal{A}}

\DeclareDocumentCommand{\autstep}{O{}}{%
  \rightarrow_{#1}
}
\DeclareDocumentCommand{\autsteps}{O{}}{%
  \rightarrow^*_{#1}
}
\DeclareDocumentCommand{\autstepp}{O{} m}{%
  \xrightarrow{#2}_{#1}
}

\newcommand{\rdown}[1]{\overrightarrow{#1}}
\newcommand{\ldown}[1]{\overleftarrow{#1}}

\newtheorem{theorem}{Theorem}
\newtheorem{lemma}[theorem]{Lemma}

\Crefname{theorem}{Theorem}{Theorems}
\Crefname{lemma}{Lemma}{Lemmas}
\Crefname{fact}{Fact}{Facts}
\Crefname{proposition}{Proposition}{Propositions}

\begin{document}
\begin{abstract}
  We consider the model of pushdown vector addition systems with
  resets.  These consist of vector addition systems that have access
  to a pushdown stack and have instructions to reset counters. For
  this model, we study the coverability problem.  In the absence of
  resets, this problem is known to be decidable for one-dimensional
  pushdown vector addition systems, but decidability is open for
  general pushdown vector addition systems.  Moreover, coverability is
  known to be decidable for reset vector addition systems without a
  pushdown stack.  We show in this note that the problem is
  undecidable for one-dimensional pushdown vector addition systems
  with resets.

  \keywords Pushdown vector addition systems; decidability
\end{abstract}
\maketitle

\input{body}

\section*{Acknowledgements}

Work partially funded by
\href{http://bravas.labri.fr/}{ANR-17-CE40-0028~\textsc{Bra\!VAS}}.

\printbibliography

\end{document}

%% file: drawings.tex
\newcommand{\resetdistance}{1.0cm}

\tikzset{gadget/.style={->,>=stealth,initial text=,minimum size=0pt,auto,on grid,scale=1,inner sep=1pt,node distance=2.8cm}}
\tikzset{every state/.style={minimum size=0pt,inner sep=1pt,fill=black!10,draw=black!70,thick}}

\newcommand{\gadgetmultiplication}{
\begin{tikzpicture}[gadget]
  \node[state] (q1) [initial] {$q_1$};
  \node[state] (q2) [right=of q1] {$q_{2}$};
  \node[state] (q3) [accepting by arrow,right=\resetdistance of q2] {$q_3$};
  \path (q1) edge [loop above] node {$\bar{a}\iinc^f$} (q1)
        (q1) edge  node {$\bar{\#}m_f\#$} (q2);
  \path (q2) edge [loop above] node {$\idec a$} (q2)
        (q2) edge node {$\ireset$} (q3);
\end{tikzpicture}
}

\newcommand{\gadgetinversemultiplication}{
\begin{tikzpicture}[gadget]
  \node[state] (q1) [initial] {$q_1$};
  \node[state] (q2) [right=of q1] {$q_{2}$};
  \node[state] (q3) [accepting by arrow, right=\resetdistance of q2] {$q_3$};
  \path (q1) edge [loop above] node {$\bar{a}^f\iinc$} (q1)
        (q1) edge  node {$\bar{\#}\bar{m}_f\#$} (q2);
  \path (q2) edge [loop above] node {$\idec a$} (q2)
        (q2) edge node {$\ireset$} (q3);
\end{tikzpicture}
}

\newcommand{\gadgetdivision}{
\begin{tikzpicture}[gadget]
  \node[state] (q1) [initial] {$q_1$};
  \node[state] (q2) [right=of q1] {$q_{2}$};
  \node[state] (q3) [accepting by arrow, right=\resetdistance of q2] {$q_3$};
  \path (q1) edge [loop above] node {$\bar{a}^f\iinc$} (q1)
        (q1) edge  node {$\bar{\#}d_f\#$} (q2);
  \path (q2) edge [loop above] node {$\idec a$} (q2)
        (q2) edge node {$\ireset$} (q3);
\end{tikzpicture}
}

\newcommand{\gadgetinversedivision}{
\begin{tikzpicture}[gadget]
  \node[state] (q1) [initial] {$q_1$};
  \node[state] (q2) [right=of q1] {$q_{2}$};
  \node[state] (q3) [accepting by arrow, right=\resetdistance of q2] {$q_3$};  
  \path (q1) edge [loop above] node {$\bar{a}\iinc^f$} (q1)
        (q1) edge  node {$\bar{\#}\bar{d}_f\#$} (q2);
  \path (q2) edge [loop above] node {$\idec a$} (q2)
        (q2) edge node {$\ireset$} (q3);
\end{tikzpicture}
}

\newcommand{\gadgettest}{
\begin{tikzpicture}[gadget]
  \node[state] (q1) [initial] {$q_1$};
  \node[state] (q2) [right=of q1] {$q_{2}$};
  \node[state] (q3) [accepting by arrow, right=\resetdistance of q2] {$q_3$};
  \path (q1) edge [loop above] node {$\bar{a}^f\iinc^f$} (q1)
        (q1) edge  node {$\bar{a}^g\iinc^g\bar{\#}t_f\#$} (q2);
  \path (q2) edge [loop above] node {$\idec a$} (q2)
        (q2) edge node {$\ireset$} (q3);
\end{tikzpicture}
}

\newcommand{\gadgetinversetest}{
\begin{tikzpicture}[gadget]
  \node[state] (r1) [initial] {$q_1$};
  \node[state] (r2) [right= of r1] {$q_2$};
  \node[state] (r3) [accepting by arrow, right=\resetdistance of r2] {$q_3$};

  \path (r1) edge [loop above] node {$\bar{a}^f\iinc^f$} (r1)
        (r1) edge node {$\bar{a}^g\iinc^g\bar{\#}\bar{t}_f\#$} (r2);
  \path (r2) edge [loop above] node {$\idec a$} (r2)
        (r2) edge node {$\ireset$} (r3);
\end{tikzpicture}
}

\newcommand{\extrastates}{
\begin{tikzpicture}[gadget, node distance=2.1cm]
  \node[state] (s') [initial] {$s'$};
  \node[state] (s) [right= of s'] {$s$};
  \node[state] (t) [right= of s] {$t$};
  \node[state] (b) [right= of t] {$b$};
  \node[state] (t') [accepting by arrow, right= 2.5cm of b] {$t'$};

  \node[ rectangle] (mf) [above=1cm of b] {$\ainvmult{f}$%
  };
  \node[ rectangle] (df) [above=0.7cm of mf] {$\ainvdiv{f}$%
  };
  \node[ rectangle] (tf) [above=0.7cm of df] {$\ainvtest{f}$%
  };

  \draw[dashed,rounded corners=6pt,color=black!70,thick] ($(s.west)+(-0.2,-0.5)$) rectangle ( $(t.east)+(0.2,0.5)$ ); 

  \path (s') edge node {$\bot \# a$} (s);
  \path (t) edge node {$\bar{a}\bar{\#}\#a$} (b);
  \path (b) edge node {$\bar{a}\bar{\#}\bar{\bot}$} (t');

  \path (b) edge [out=-225, in=-180] (mf.west);
  \path (b) edge [out=-225, in=-180] (df.west);
  \path (b) edge [out=-225, in=-180] (tf.west);
  
  \path (mf.east) edge [out=0, in=45] (b);
  \path (df.east) edge [out=0, in=45] (b);
  \path (tf.east) edge [out=0, in=45] (b);

\end{tikzpicture}
}

%% file: body.tex
\section{Introduction}

Vector addition systems with states (VASS) play a central role for
modelling systems that manipulate discrete resources, and as such
provide an algorithmic toolbox applicable in many different fields.
Adding a pushdown stack to vector addition systems yields so-called
\emph{pushdown VASS} (PVASS), which are even more versatile: one can
model for instance recursive programs with integer
variables~\cite{atig11} or distributed systems with a recursive server
and multiple finite-state clients, and PVASS can be related to
decidability issues in logics on data trees~\cite{lazic13}.  However,
this greater expressivity comes with a price: the \emph{coverability
  problem} for PVASS is only known to be decidable in dimension
one~\cite{leroux2015coverability}.  This problem captures most of the
decision problems of interest and in particular safety properties, and
is the stumbling block in a classification for a large family of
models combining pushdown stacks and counters~\cite{zetzsche15}.

Another viewpoint on one-dimensional
PVASS~\cite{leroux2015coverability} is to see those systems as
extensions of two-dimensional VASS, where one of the two counters is
replaced by a pushdown stack.  In this context, a complete
classification with respect to decidability of coverability, and of
the more difficult \emph{reachability problem}, was provided by Finkel
and Sutre~\cite{FinkelSutre2000}, 
\begin{table}[th]
\centering
\caption{Decidability status of the coverability and reachability problems in extensions of two-dimensional VASS; our contribution is indicated in bold.\label{tab-decidability}}
    \newcolumntype{Y}{>{\raggedright\arraybackslash}X}
    \definecolor{llg}{cmyk}{0,0,0,0.04}
    \begin{subtable}{0.45\textwidth}
    \centering
    \caption{Coverability problem.\label{tab-cov}}
    \begin{tabularx}{\textwidth}{YYYYY|Y}
      $\Nctr$         & $\NctrReset$         & $\NctrTransfer$         & $\NctrZero$             & $\PD$                           &\\
      \hline
      D~\cite{karp69} & D~\cite{arnold78}    & D~\cite{dufourd98}      & D~\cite{reinhardt08}    & D~\cite{leroux2015coverability} & $\Nctr$\\
      \cellcolor{llg} & D~\cite{arnold78}    & D~\cite{dufourd98}      & D~\cite{FinkelSutre2000}& \textbf{U}                      & $\NctrReset$ \\
      \cellcolor{llg} & \cellcolor{llg}      & D~\cite{dufourd98}      & U~\cite{FinkelSutre2000}& U~\cite{FinkelSutre2000}        & $\NctrTransfer$ \\
      \cellcolor{llg} & \cellcolor{llg}      & \cellcolor{llg}         & U~\cite{Minsky1961}     & U~\cite{Minsky1961}             & $\NctrZero$ \\
      \cellcolor{llg} & \cellcolor{llg}      & \cellcolor{llg}         & \cellcolor{llg}         & U                               & $\PD$
    \end{tabularx}
    \end{subtable}\quad
    \begin{subtable}{0.45\textwidth}
    \caption{Reachability problem.\label{tab-reach}}
    \centering
    \begin{tabularx}{\textwidth}{YYYYY|Y}
      $\Nctr$           & $\NctrReset$            & $\NctrTransfer$         & $\NctrZero$             & $\PD$                           &\\
      \hline
      D~\cite{leeuwen74}& D~\cite{reinhardt08}    & D~\cite{reinhardt08}    & D~\cite{reinhardt08}    & ??                              & $\Nctr$\\
      \cellcolor{llg}   & D~\cite{FinkelSutre2000}& D~\cite{FinkelSutre2000}& D~\cite{FinkelSutre2000}& \textbf{U}                      & $\NctrReset$ \\
      \cellcolor{llg}   & \cellcolor{llg}         & D~\cite{FinkelSutre2000}& U~\cite{FinkelSutre2000}& U~\cite{FinkelSutre2000}        & $\NctrTransfer$ \\
      \cellcolor{llg}   & \cellcolor{llg}         & \cellcolor{llg}         & U~\cite{Minsky1961}     & U~\cite{Minsky1961}             & $\NctrZero$ \\
      \cellcolor{llg}   & \cellcolor{llg}         & \cellcolor{llg}         & \cellcolor{llg}         & U                               & $\PD$
    \end{tabularx}
\end{subtable}
\end{table}%
 whether one uses plain counters
($\Nctr$), counters with resets~($\NctrReset$), counters whose
contents can be transferred to the other counter~($\NctrTransfer$), or
counters with zero tests~($\NctrZero$); see \cref{tab-decidability}.
In particular, two-dimensional VASS with one counter extended to allow
resets and one extended to allow zero tests have a decidable
reachability problem~\cite{FinkelSutre2000}: put differently, the
coverability problem for one-dimensional PVASS \emph{with resets}
($1$-PRVASS) is decidable if the stack alphabet is of the
form~$\{a,\bot\}$ where~$\bot$ is a distinguished bottom-of-stack
symbol.

\subsubsection*{Contributions.}
In this note, we show that Finkel and Sutre's decidability result does
not generalise to one-dimensional pushdown VASS with resets over an
arbitrary finite stack alphabet.
\begin{theorem}\label{thm}
  The coverability problem for $1$-PRVASS is undecidable.
\end{theorem}
\noindent
As far as the coverability problem is concerned, this fully determines
the decidability status in extensions of two-dimensional VASS where
one may also replace counters by pushdown stacks~($\PD$); see
\cref{tab-cov}.

Technically, the proof of \cref{thm} presented in \cref{sec-red}
reduces from the reachability problem in two-counter Minsky machines.
The reduction relies on the ability to \emph{weakly
  implement}~\cite{mayr81} basic operations---like multiplication by a
constant---and their inverses---like division by a constant.  This in
itself would not bring much; for instance, plain two-dimensional VASS
can already weakly implement multiplication and division by constants.
The crucial point here is that, in a \mbox{$1$-PRVASS}, we can also
weakly implement the inverse of a \emph{sequence} of basic operations
performed by the system, by using the pushdown stack to record a sequence
of basic operations and later replaying it in reverse, and relying on
resets to ``clean-up'' between consecutive operations.  Note that
without resets, while PVASS are known to be able to weakly implement
Ackermannian functions already in dimension one~\cite{leroux14}, they
cannot weakly compute sublinear functions~\cite{leroux19}---like
iterated division by two, i.e., logarithms.

\section{Pushdown Vector Addition Systems with Resets}\label{sec-pvassr}
A \emph{($1$-dimensional) pushdown vector addition system with resets
  ($1$-PRVASS)} is a tuple $\cV=(Q,\Gamma,A)$, where $Q$ is a finite
set of \emph{states}, $\Gamma$ is a finite set of \emph{stack symbols}, and
$A\subseteq Q\times I^*\times Q$ is a finite set of
\emph{actions}.  Here, transitions are labelled by finite
sequences of \emph{instructions} from
$I\eqdef\Gamma\cup\bar{\Gamma}\cup \{\iinc,\idec,\ireset\}$ where
$\bar\Gamma\eqdef\{\bar z\mid z\in\Gamma\}$ is a disjoint copy
of~$\Gamma$.

\medskip
A $1$-PRVASS defines a (generally infinite) transition system acting
over \emph{configurations} $(q,w,n)\in Q\times\Gamma^*\times\N$.  For
an instruction~$x\in I$, $w,w'\in\Gamma^\ast$, and $n,n'\in\+N$, we
write $(w,n)\autstepp{y}(w',n')$ in the following cases:
\begin{description}
\item[push] if $x=z$ for $z\in\Gamma$, then $w'=wz$ and $n'=n$,
\item[pop] if $x=\bar{z}$ for $z\in\Gamma$, then $w=w'z$ and $n'=n$,
\item[increment] if $x=\iinc$, then $w'=w$ and $n'=n+1$.
\item[decrement] if $x=\idec$, then $w'=w$ and $n'=n-1$, and
\item[reset] if $x=\ireset$, then $w'=w$ and $n'=0$.
\end{description}
Moreover, for a sequence of instructions $u=x_1\cdots x_k$ with
$x_1,\ldots,x_k\in I$, we have $(w_0,n_0)\autstepp{u}(w_k,n_k)$ if
for some $(w_1,n_1),\ldots,(w_{k-1},n_{k-1})\in\Gamma^*\times\N$, we have
$(w_i,n_i)\autstepp{x_i}(w_{i+1},n_{i+1})$ for all $0\leq i<k$.
Finally, for two configurations
$(q,w,n),(q',w',n')\in Q\times\Gamma^*\times\N$, we write
$(q,w,n)\autstep[\cV](q',w',n')$ if there is an action $(q,u,q')\in A$
such that $(w,n)\autstepp{u}(w',n')$.

\medskip
The \emph{coverability problem} for 1-PRVASS is the following decision
problem.
\begin{description}
\item[given] a 1-PRVASS $\cV=(Q,\Gamma,A)$, states $s,t\in Q$.
\item[question] are there $w\in\Gamma^*$ and $n\in\N$ with
  $(s,\varepsilon,0)\autsteps[\cV](t,w,n)$?
\end{description}

\section{Reduction from Minsky Machines}\label{sec-red}

\newcommand{\amult}[1]{\mathcal{M}_{#1}}
\newcommand{\ainvmult}[1]{\bar{\mathcal{M}}_{#1}}
\newcommand{\adiv}[1]{\mathcal{D}_{#1}}
\newcommand{\ainvdiv}[1]{\bar{\mathcal{D}}_{#1}}
\newcommand{\atest}[1]{\mathcal{T}_{#1}}
\newcommand{\ainvtest}[1]{\bar{\mathcal{T}}_{#1}}

We present in this section a reduction from reachability in
two-counter Minsky machines to coverability in 1-PRVASS.

\subsection{Preliminaries}
Recall that a \emph{two-counter (Minsky) machine} is a tuple
$\cM=(Q,A)$, where $Q$ is a finite set of states and
$A\subseteq Q\times \{0,1\}\times \{\iinc,\idec,\izero\}\times Q$ a
set of actions.  A \emph{configuration} is a now triple $(q,n_0,n_1)$
with $q\in Q$ and $n_0,n_1\in\N$.  We write
$(q,n_0,n_1)\autstep[\cM](q',n'_0,n'_1)$ if there is an action
$(q,c,x,q')\in A$ such that $n'_{1-c}=n_{1-c}$ and
\begin{description}
\item[increment] if $x=\iinc$, then $n'_c=n_c+1$,
\item[decrement] if $x=\idec$, then $n'_c=n_c-1$, and
\item[zero test] if $x=\izero$, then $n'_c=n_c=0$.
\end{description}
The \emph{reachability problem} for two-counter machines is the
following undecidable decision problem~\cite{Minsky1961}.
\begin{description}
\item[given] a two-counter machine $\cM=(Q,A)$, and states $s,t\in Q$.
\item[question] does $(s,0,0)\autsteps[\cM](t,0,0)$ hold?
\end{description}

\subsubsection*{G\"odel Encoding.}
The first ingredient of the reduction is to use the well-known
encoding of counter values $(n_0,n_1)\in\N\times\N$ as a single number
$2^{n_0}3^{n_1}$; for instance, the pair $(0,0)\in\N\times\N$ is
encoded by $2^03^0=1$.  In this encoding, incrementing the first
counter means multiplying by~$2$, decrementing the second counter
means dividing by~$3$, and testing the second counter for zero means
verifying that the encoding is not divisible by~$3$, etc.  Note that,
in each case, we encode the instruction as a partial function
$g\colon \N\nrightarrow\N$; let us we define its \emph{graph} as the
binary relation $R\eqdef\{(m,n)\in\N\times\N \mid \text{$g$ is defined
  on $m$ and $g(m)=n$}\}$.  Thus the encoded instructions are the
partial functions with the following graphs:
\begin{align*}
  {\mathbin{R_{m_f}}}&\eqdef\{(n,f\cdot n) \mid n\in\N\} & &\text{for multiplication,}\\
  R_{d_f}&\eqdef\{(f\cdot n, n) \mid n\in\N\} & & \text{for division, and} \\
  R_{t_f}&\eqdef\{(n,n) \mid n\not\equiv 0\bmod{f} \} & & \text{for the divisibility test,}
\end{align*}
for a factor $f\in\{2,3\}$.  This means that we can equivalently see
\begin{itemize}
\item a two-counter machine with distinguished source and target
  states~$s$ and~$t$ as a regular language $M\subseteq\Delta^*$ over
  the alphabet $\Delta\eqdef\{m_f, d_f, t_f\mid f\in\{2,3\}\}$, and
\item reachability as the existence of a word $u=x_1\cdots x_\ell$ in
  the language~$M$, with $x_1,\dots,x_\ell\in\Delta$, such that the pair
  $(1,1)$ belongs to the composition
  $R_{x_1}R_{x_2}\cdots R_{x_\ell}$.
\end{itemize}

\subsubsection*{Weak Relations.} Here, the problem is that it does
not seem possible to implement these operations (multiplication,
division, divisibility test) directly in a $1$-PRVASS.  Therefore, a
key idea of our reduction is to perform the instructions of $u$
\emph{weakly}---meaning that the resulting value may be smaller than
the correct result---but \emph{twice}: once forward and once backward.
More precisely, for any relation $R\subseteq \N\times\N$, we define
the \emph{weak forward} and \emph{backward} relations $\rdown{R}$ and $\ldown{R}$ by
\begin{align*}
  \rdown{R}&\eqdef\{(m,n)\in\N\times\N \mid \exists \tilde{n}\ge n\colon (m,\tilde{n})\in R\} \\
  \ldown{R}&\eqdef\{(m,n)\in\N\times\N \mid \exists \tilde{m}\ge m\colon (\tilde{m},n)\in R\}.
\end{align*}
Let us call a relation $R\subseteq\N\times\N$ \emph{strictly monotone}
if for $(m,n)\in R$ and $(m',n')\in R$, we have $m<m'$ if and only if
$n<n'$.  We shall rely on the following
\lcnamecref{two-approximations}, which is proven in
\cref{proof-two-app}.
\begin{restatable}{proposition}{twoapprox}\label{two-approximations}
  If $R_1,\ldots,R_\ell\subseteq\N\times\N$ are strictly monotone relations, then
  $R_1R_2\cdots R_\ell=\rdown{R_1}\rdown{R_2}\cdots\rdown{R_\ell}\cap
  \ldown{R_1}\ldown{R_2}\cdots\ldown{R_\ell}$.
\end{restatable}

We shall thus construct in \cref{sec-rec} a $1$-PRVASS~$\cV$ in which
a particular state is reachable if and only if there exists a
word~$u\in M$ with $u=x_1\cdots x_\ell$ and
$x_1,\ldots,x_\ell\in\Delta$, such that $(1,1)\in
\rdown{R_{x_1}}\cdots \rdown{R_{x_\ell}}$ and $(1,1)\in
\ldown{R_{x_1}}\cdots \ldown{R_{x_\ell}}$.  Since the relations
$R_{m_f}$, $R_{d_f}$, and $R_{t_f}$ for $f\in\{2,3\}$ are strictly
monotone, \cref{two-approximations} guarantees that this is equivalent
to $(1,1)\in R_{x_1}\cdots R_{x_\ell}$.  Intuitively, if we make a
mistake in the forward phase $\rdown{R_{x_1}}\cdots
\rdown{R_{x_\ell}}$, then at some point, we produce a number
$n$ that is smaller than the correct result $\tilde n> n$.
Then, the backward phase cannot compensate for that, because it can
only make the results even smaller, and cannot reproduce the initial
value.

\subsection{Construction}\label{sec-rec}

We now describe the construction of our $1$-PRVASS~$\cV$.  Its stack
alphabet $\Gamma\eqdef\Delta\cup \{\bot,\#,a\}$. In~$\cV$, each
configuration will be of the form $(q,\bot w\#a^n, k)$, where
$w\in\Delta^*$, and $n,k\in\N$.  In the forward phase, we simulate the
run of the two-counter machine so that~$n$ is the G\"{o}del encoding
of the two counters.  In order to perform the backward phase, the
word~$w$ records the instruction sequence of the forward phase.  The
resettable counter is used as an auxiliary counter in each weak
computation step.

\subsubsection*{Gadgets.}
\begin{figure}[tbp]
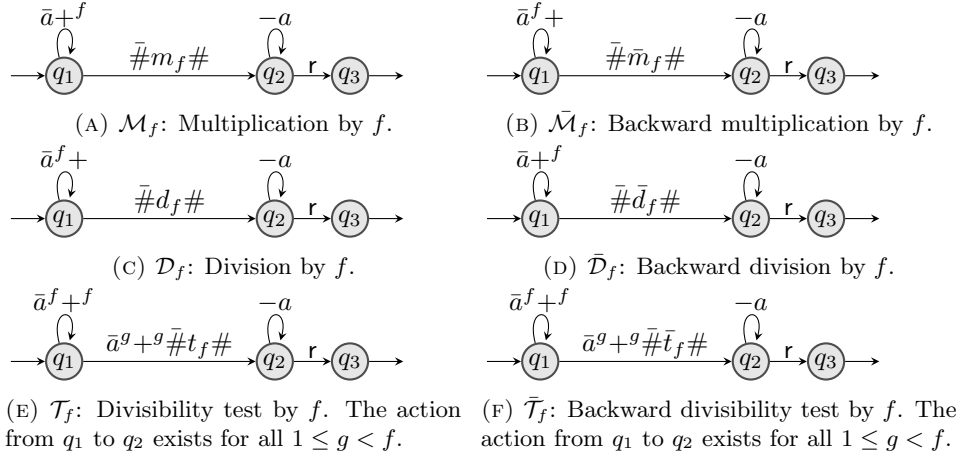

  \begin{subfigure}[t]{0.48\textwidth}
    \gadgetmultiplication
    \subcaption{$\amult{f}$: Multiplication by $f$.}
  \end{subfigure}\hfill
  \begin{subfigure}[t]{0.50\textwidth}
    \gadgetinversemultiplication
    \subcaption{$\ainvmult{f}$: Backward multiplication by $f$.}
  \end{subfigure}
  
  \begin{subfigure}[t]{0.48\textwidth}
    \gadgetdivision
    \subcaption{$\adiv{f}$: Division by $f$.}
  \end{subfigure}\hfill
  \begin{subfigure}[t]{0.50\textwidth}
    \gadgetinversedivision
    \subcaption{$\ainvdiv{f}$: Backward division by $f$.}
  \end{subfigure}
  
  \begin{subfigure}[t]{0.48\textwidth}
    \gadgettest
    \subcaption{$\atest{f}$: Divisibility test by $f$. The action
      from $q_1$ to $q_2$ exists for all $1\leq g<f$.}
  \end{subfigure}\hfill
  \begin{subfigure}[t]{0.50\textwidth}
    \gadgetinversetest
    \subcaption{$\ainvtest{f}$: Backward divisibility test by
      $f$. The action from $q_1$ to $q_2$ exists for all $1\leq g< f$.}
  \end{subfigure}
  \caption{Gadgets used in the reduction.}
  \label{gadgets}
\end{figure}
For each weak computation step, we use one of the gadgets from
\cref{gadgets}; note that, for instance, ``$+^f$'' denotes the
sequence of instructions $+\cdots +$ of length~$f$.  Observe that we
have:
\begin{align}
  (q_1,\bot u\#a^m,0)&\autsteps[\amult{f}](q_3,\bot v\#a^n,0) &
  &\text{iff} & &v=um_f \text{ and }(m,n)\in\rdown{R_{m_f}} \\
  (q_1,\bot u\#a^m,0)&\autsteps[\ainvmult{f}](q_3,\bot v\#a^n,0) &
  &\text{iff} & &u=vm_f\text{ and }(n,m)\in\ldown{R_{m_f}} %
\end{align}
and analogous facts hold for $\adiv{f}$ and $\ainvdiv{f}$ (with $d_f$
instead of $m_f$) and also for $\atest{f}$ and $\ainvtest{f}$ (with
$t_f$ instead of $m_f$).  Let us explain this in the case
$\amult{f}$. In the loop at $q_1$, $\amult{f}$ removes $a$ from the
stack and adds $f$ to the auxiliary counter. When $\#$ is on top of
the stack the automaton moves to $q_2$ and changes the stack from
$\bot u\#$ to $\bot um_f\#$.  Therefore, once $\amult{f}$ is in $q_2$,
it has set the counter to $f\cdot m$. In the loop at $q_2$, it
decrements the counter and pushes $a$ onto the stack before it resets
the counter and moves to $q_3$. Thus, in state $q_3$, we have
$0\leq n\leq f\cdot m$.

\subsubsection*{Main Control.}
Let $M\subseteq\Delta^*$ be accepted by the finite automaton
$\cA=(\Delta,Q,A,s,t)$. Schematically, our $1$-PRVASS~$\cV$ is
structured as in the following diagram:
\[ \extrastates\]
The part in the dashed rectangle is obtained from~$\cA$ as follows.
Whenever there is an action $(q,m_f,q')$ in~$\cA$, we glue in a fresh
copy of~$\amult{f}$ between~$q$ and~$q'$, including
$\varepsilon$-actions from~$q$ to~$q_1$ and from~$q_3$ to~$q'$.  The
original action $(q,m_f,q')$ is removed.  We proceed analogously for
actions $(q,d_f,q')$ and $(q,t_f,q')$, where we glue in fresh copies
of~$\adiv{f}$ and~$\atest{f}$, respectively.  Clearly, the part in the
dashed rectangle realizes the forward phase as described above.

Once it reaches~$t$, $\cV$ can check if the current number stored on
the stack equals~$1$ and if so, move to state~$b$.  In state~$b$, the
backward phase is implemented.  The $1$-PRVASS~$\cV$ contains a copy of
$\ainvmult{f}$, $\ainvdiv{f}$, and $\ainvtest{f}$ for each
$f\in\{2,3\}$.  Each of these copies can be entered from~$b$ and goes
back to~$b$ when exited.

Finally, the stack is emptied by an action from~$b$ to~$t'$, which can
be taken if and only if the stack content is~$\bot\# a$.  We can check
that from $(s',\varepsilon,0)$, one can reach a configuration
$(t',w,m)$ with $w\in\Gamma^*$ and $m\in\N$, if and only if there exists
$u\in M$, $u=x_1\cdots x_\ell$, and $x_1,\ldots,x_\ell\in\Delta$, with
$(1,1)\in \rdown{R_{x_1}}\cdots \rdown{R_{x_\ell}}\cap
\ldown{R_{x_1}}\cdots \ldown{R_{x_\ell}}$.
According to \cref{two-approximations}, the latter is equivalent to
$(1,1)\in R_{x_1}\cdots R_{x_\ell}$.

\section{Concluding Remarks}

In this note, we have proven the undecidability of coverability in
one-dimensional pushdown VASS with resets (c.f.\ \cref{thm}).  The
only remaining open question in \cref{tab-decidability} regarding
extensions of two-dimensional VASS is a long-standing one, namely the
reachability problem for one-dimensional PVASS.  Another fruitful
research avenue is to pinpoint the exact complexity in the decidable
cases of \cref{tab-decidability}. Here, not much is known except
regarding coverability and reachability in two-dimensional VASS: these
problems are \PSPACE-complete if counter updates are encoded in
binary~\cite{blondin15} and \NL-complete if updates are encoded in
unary~\cite{EnglertLT16}.

\appendix

\newcommand{\vge}{\rotatebox[origin=c]{-90}{$\ge$}}
\newcommand{\vgt}{\rotatebox[origin=c]{-90}{$>$}}
\newcommand{\veq}{\rotatebox[origin=c]{-90}{$=$}}

\section{Proof of \Cref{two-approximations}}\label{proof-two-app}
It remains to prove \cref{two-approximations}.  We will use the
following lemma.

\begin{lemma}\label{monotone-pairs}
  Let $R_1,\ldots,R_\ell\subseteq\N\times\N$ be strictly monotone
  relations and $(m,n)\in \rdown{R_1}\cdots\rdown{R_\ell}$ and
  $(m',n')\in \ldown{R_1}\cdots\ldown{R_\ell}$.  If $n'\le n$, then
  $m'\le m$. Moreover, if $n'<n$, then $m'<m$.
\end{lemma}
\begin{proof}
  It suffices to prove the \lcnamecref{monotone-pairs} in the case
  $\ell=1$: then, the general version follows by induction.  Let
  $(m,n)\in \rdown{R_1}$ and $(m',n')\in\ldown{R_1}$.  Then there are
  $\tilde{n}\ge n$ with $(m,\tilde{n})\in R_1$ and $\tilde{m}\ge m'$
  with $(\tilde{m},n')\in R_1$.  If $n'<n$, then we have the following
  relationships:
  \[
  \begin{matrix}
    m & R_1 & \tilde{n} \\
    &   & \vge      \\
    &   & n \\
    &   & \vgt \\
    \tilde{m}& R_1  & n' \\
    \vge   &   & \\
    m'     &    & 
  \end{matrix} 
\]
Since $R_1$ is strictly monotone, this implies $\tilde{m}<m$ and thus $m'<m$. The case
$n'\le n$ follows by the same argument.\myqed
\end{proof}

We are now ready to prove \cref{two-approximations}.

\twoapprox*
\begin{proof}
  Of course, for any relation $R\subseteq\N\times\N$, one has
  $R\subseteq \rdown{R}$ and $R\subseteq\ldown{R}$. In particular,
  $R_1R_2\cdots R_\ell$ is included in both
  $\rdown{R_1}\rdown{R_2}\cdots\rdown{R_\ell}$ and
  $\ldown{R_1}\ldown{R_2}\cdots\ldown{R_\ell}$.

  For the converse inclusion, suppose
  $(m,n)\in \rdown{R_1}\rdown{R_2}\cdots\rdown{R_\ell}\cap
  \ldown{R_1}\ldown{R_2}\cdots\ldown{R_\ell}$. Then there are
  $p_0,\ldots,p_{\ell}\in\N$ with $p_0=m$, $p_\ell=n$, and
  $(p_{i-1},p_{i})\in \rdown{R_i}$ for $0< i\leq\ell$.  There are also
  $q_0,\ldots,q_{\ell}\in\N$ with $q_0=m$, $q_\ell=n$, and
  $(q_{i-1},q_{i})\in \ldown{R_i}$ for $0<i\leq\ell$.
  
  Towards a contradiction, suppose that $(p_{i-1},p_i)\notin R_i$ for some $0<i\leq\ell$.
  Then there is a $\tilde{p}_i>p_i$ with $(p_{i-1},\tilde{p}_i)\in R_i$. With this, we
  have
  \[  \begin{matrix}
      m=p_0 & \rdown{R_1}\cdots \rdown{R_i}& \tilde{p}_i & & \\
      &                                    & \vgt        & & \\
      &                                    & p_i & \rdown{R_{i+1}}\cdots \rdown{R_\ell} & p_\ell\\
      &                                    &     &                                   & \veq \\
      m=q_0 &  \ldown{R_1}\cdots\ldown{R_i}   & q_i & \ldown{R_{i+1}}\cdots\ldown{R_\ell}  & q_\ell
    \end{matrix}
  \]
  Since $p_\ell=q_\ell$, \cref{monotone-pairs} applied to
  $R_{i+1},\dots,R_\ell$ implies $q_i\le p_i$ and thus
  $q_i<\tilde{p}_i$.  Applying \cref{monotone-pairs} to
  $R_1,\dots,R_i$ then yields $q_0<p_0$, a contradiction.  Therefore, we
  have $(p_{i-1},p_i)\in R_i$ for every $0<i\leq\ell$ and thus
  $(m,n)\in R_1\cdots R_\ell$.\myqed
\end{proof}